\documentclass[onecolumn,draftcls]{IEEEtran}

 \usepackage{subfigure}
\usepackage{epsfig,graphicx,psfrag,dsfont}
\usepackage{amsfonts,amsmath,amssymb,color}
\usepackage{bbm,setspace,amsthm,cite}
\newcommand{\Xbbf}{{\mathbb{\mathbf{X}}}}
\newcommand{\Zbbf}{{\mathbb{\mathbf{Z}}}}
\newcommand{\Lp}[3]{{\left\|{#1}\right\|}_{#2}^{#3}}

\newcommand{\LPr}[1]{\left(#1\right)}

\DeclareMathOperator*{\argmin}{arg\,min}

\usepackage{xspace}
\usepackage{bbm}

\newcommand{\Ac}{\mathcal{A}}
\newcommand{\Bc}{\mathcal{B}}
\newcommand{\Cc}{\mathcal{C}}

\newcommand{\Ec}{\mathcal{E}}
\newcommand{\Fc}{\mathcal{F}}

\newcommand{\Nc}{\mathcal{N}}

\newcommand{\Uc}{\mathcal{U}}

\newcommand{\Xc}{\mathcal{X}}
\newcommand{\Yc}{\mathcal{Y}}
\newcommand{\Zc}{\mathcal{Z}}

\newcommand{\Xh}{{\hat{X}}}

\newcommand{\xh}{{\hat{x}}}

\newcommand{\Xt}{{\tilde{X}}}

\def\d{\delta}
\def\e{\epsilon}

\def\l{\lambda}

\let\P\relax
\DeclareMathOperator\P{P}

\newcommand\ie{i.e.,\xspace}
\def\textiid{i.i.d.\@\xspace}
\newcommand\iid{\ifmmode\text{ i.i.d. } \else \textiid \fi}

\newcommand{\ind}{\mathbbmss{1}}

\newcommand{\ex}{{\rm e}}

\newcommand{\wb}{{\bf w} }
\newcommand{\Wb}{{\bf W} }

\newtheorem{question}{Question}

\newtheorem{definition}{Definition}

\newtheorem{theorem}{Theorem}
\newtheorem{lemma}{Lemma}

\newtheorem{remark}{Remark}

\begin{document}

\title{ Theoretical links between   universal and Bayesian compressed sensing algorithms}

\author{Shirin Jalali\thanks{Shirin Jalali is with the Mathematics and Algorithms Group in  Nokia - Bell Labs, New Providence, NJ,  e-mail: Shirin.Jalali@nokia-bell-labs.com}}


\date{} 

\maketitle

\begin{abstract}
Quantized maximum a posteriori (Q-MAP)  is a recently-proposed Bayesian compressed sensing algorithm that, given the source distribution, recovers  $X^n$ from its  linear measurements  $Y^m=AX^n$, where  $A\in\mathds{R}^{m\times n}$   denotes the known measurement matrix. On the other hand,  Lagrangian minimum entropy pursuit (L-MEP) is a universal compressed sensing algorithm that aims at recovering $X^n$ from its  linear measurements  $Y^m=AX^n$, without having access to the source distribution. Both Q-MAP and L-MEP provably achieve the minimum required sampling rates, in noiseless cases where such fundamental limits are known.  L-MEP is based on minimizing a cost function that consists of a linear combination of the conditional empirical entropy of a potential reconstruction vector and its corresponding measurement error. In this paper, using a first-order linear approximation of the conditional empirical entropy function, L-MEP is connected with Q-MAP.  The established connection between L-MEP and  Q-MAP leads to variants of Q-MAP which have  the same asymptotic performance as Q-MAP  in terms of their required sampling rates.  Moreover, these  variants  suggest that  Q-MAP is robust to small error in estimating the source distribution. This robustness  is theoretically proven and the effect of a non-vanishing estimation error on the required sampling rate is characterized. 

\end{abstract}

\section{Introduction}

Consider the classic problem of  compressed sensing: a high-dimensional ``structured'' signal $X^n\in\mathds{R}^n$ is to be recovered from its under-sampled linear measurements $Y^m=AX^n+Z^m$, where $m<n$, $A\in\mathds{R}^{m\times n}$  and $Z^m\in\mathds{R}^{m}$. It is crucial  for the signal $X^n$ to be structured, as otherwise the problem has infinitely many solutions. Classical compressed sensing is mainly concerned with the case where the structure of the source is a known simple structure, such as sparsity in some transform domain or low-rankness.  The problem with these typical underlying assumptions is that in many real applications 
\begin{enumerate}
\item[i)] the underlying structure (or  distribution) of the source is not known (or is only  partially known),
\item[ii)] the source structure is much more complicated than  the typical structures  that are currently considered in compressed sensing. 
\end{enumerate} 
Recovery algorithms that do not address these issues suffer from sub-optimal recovery performance. This means that they  require more measurements  (or  have lower reconstruction quality) compared to an  algorithm that takes these issues into account.

The first issue, i.e., not knowing the source distribution or structure, is an important problem that happens in many other data processing tasks, such as compression, denoising and prediction. To address this issue, researchers in different fields have  looked into   designing   ``universal'' algorithms that do not require  knowledge of the source distribution. In information theory,  an algorithm is called   universal  with respect to  a  class of distributions\footnote{Typically this is the class of all stationary and ergodic processes with a common  finite alphabet.}, if, without knowing the source distribution,  it asymptotically achieves the optimal   performance  for all signals that are generated according to one of the distributions in that specific class  \cite{cover}. Existence of efficient universal algorithms is already  known for a number of tasks such as lossless data compression \cite{LZ} and discrete denoising \cite{dude}.    The problem of universal compressed sensing was originally  introduced in \cite{JalaliM:14-MEP-IT} for non-stochastic sources.  Later, \cite{BaDu11,ZhuB:15,JalaliP:17-IT}  studied the problem of universal compressed sensing of stochastic processes. They proposed an implementable yet a computationally-demanding universal algorithm  for compressed sensing. The algorithm was referred to as Lagrangian minimum entropy pursuit (L-MEP) in \cite{JalaliM:14-MEP-IT}. For stationary memoryless sources that are distributed according to a mixture of discrete and continuous distributions,  L-MEP can be proved to achieve the minimum required sampling rate  \cite{JalaliM:14-MEP-IT}. 
 
 Addressing the second issue, i.e., designing efficient and robust \emph{Bayesian}  compressed sensing algorithms that  take advantage of  known complicated structures of a source is more specific to compressed sensing and has already been addressed to a great extent in other areas such as denoising and data compression.  
 One potential path to address this issue is to transfer knowledge from  other areas  and  design compression-based \cite{beygi2017efficient} or denoising-based \cite{DAMP} compressed sensing recovery algorithms. This approach turns out to be very effective and   yields  algorithms that  achieve state-of-the-art performance in, for example,  compressive imaging.   However, this method does not directly solve the Bayesian compressed problem in which the source distribution is known or learned from a large training dataset. For any  given  source distribution, a compression-based  (denoising-based) approach requires access to an  efficient compression (denoising) algorithm tailored for that known distribution. But such compression or denoising algorithm might not be available for  a general distribution of interest. Moreover, even in cases where efficient compression/denoising algorithms exist, given a big dataset, one can potentially go beyond them, by first learning all  complicated structures of the source and then designing a recovery algorithm that takes advantage of the learned model.  
 
 Quantized maximum a posteriori (Q-MAP)  is a recently-proposed Bayesian compressed sensing algorithm that aims at addressing the just-described issue \cite{jalali2016qmap}. Consider $X^n$ drawn from a general  distribution.  Q-MAP recovers  $X^n$ from its noisy  linear measurements  $Y^m=AX^n+Z^m$ by taking advantage of the source full structure captured by its distribution. 
 
Both L-MEP and Q-MAP  entail solving    discrete optimizations over the space of quantized reconstruction sequences. Let $\Xc$ denote the  alphabet of the source, which is typically an interval in $\mathds{R}$.  Then, we define $\Xc_b$ to denote the discrete set derived by quantizing all elements in $\Xc$ into $b$ bits. (The exact quantization operation will be defined later in the notation section.)  Using this definition, L-MEP and Q-MAP both try to minimize their respective cost functions over $\Xc_b^n$, for some appropriately chosen quantization  level $b$.   In  the case of L-MEP, the optimization can be written as
\begin{align}
\Xh^n_{\rm LMEP}=\argmin_{u^n\in\Xc_b^n} \LPr{f_1(u^n)+\lambda \Lp{Au^n-Y^m}{2}{2}},\label{eq:lmep-f1}
\end{align}
where $\lambda>0$ and $f_1(u^n)$ is a function, to be defined explicitly  in Section \ref{eq:from-lmep-to-qmap}, that estimates the level of structuredness of candidate reconstruction sequence  $u^n$. Note that L-MEP is a universal algorithm and  $f_1$ is a generic function that does not depend on the source distribution.  On the other hand, the optimization required  in  Q-MAP can be written as 
\begin{align}
\Xh^n_{\rm QMAP}=\argmin_{u^n\in\Xc_b^n} \LPr{f_2(u^n)+\lambda \Lp{Au^n-Y^m}{2}{2}},\label{eq:qmap-f2}
\end{align}
where $\lambda>0$ and $f_2(u^n)$, to be defined explicitly  in Section \ref{eq:from-lmep-to-qmap}, is a function that depends on the \emph{known} source distribution. Comparing \eqref{eq:lmep-f1} and \eqref{eq:qmap-f2} reveals that the only difference between L-MEP and Q-MAP  is in how they enforce the source model (structure). Interestingly, as shown in \cite{JalaliM:14-MEP-IT} and \cite{jalali2016qmap},  in the noiseless setting, asymptotically, both optimizations can recover $X^n$ from its  samples $Y^m=AX^n$, as long as the sampling rate ($m/n$) is larger than a rate that seems to be the fundamental limit for almost lossless recovery. This seems counterintuitive, as one of them (Q-MAP) is using the full source distribution and the other one (L-MEP) is a universal algorithm.  The first objective of this paper is to understand why two seemingly different optimizations have the same asymptotic performance and require the same number of measurements.   Towards this goal, we show how using a first-order approximation of function $f_1$, L-MEP reduces to an optimization similar to Q-MAP.  This approximation  not only shows a fundamental connection between Q-MAP and L-MEP, but also  enables  us to derive   variants of Q-MAP.  While these variants all have the same asymptotic performance as Q-MAP and L-MEP, each uses a different cost function to  measure the level of  structuredness  of signals in $\Xc_b^n$. 

One advantage of the variants of Q-MAP derived from L-MEP is that they provide a roadmap into studying important properties of  Q-MAP, such as its robustness to source estimation error. As mentioned earlier, Q-MAP is a Bayesian compressed sensing recovery method and requires knowing the source distribution. More precisely,  the function $f_2$  in \eqref{eq:qmap-f2} that captures the source distribution involves  some weights that are a function of  the source distribution. However, in practice,  the input distribution is rarely known in advance and typically has to be learned from some  training dataset.  This raises an important question regarding  robustness of Q-MAP to error in estimating the source distribution.   To address this issue, we first note that  in one of  the variants of Q-MAP derived from L-MEP, the   weights in function $f_2$, instead of the source distribution, are a function of the empirical distribution  of $[X^n]_b$.  Therefore,  this variant of Q-MAP already employs estimated weights instead of the ideal weights required by Q-MAP. Using this result as a first step, in this paper, we  show that in general Q-MAP is in fact robust to small estimation error. We  also characterize the required increase in the sampling rate to compensate for a non-vanishing estimation error.  

We should mention  that  since L-MEP and Q-MAP  involve solving  discrete optimizations over the space of high-resolution quantized reconstruction sequences, they are  both computationally demanding algorithms. To address the computational complexity issue, the authors in \cite{jalali2016qmap} propose an iterative algorithm  based on projected gradient descent that aim at approximating the solution of the Q-MAP algorithm. This approach partially addresses the computational-complexity issue.   

The organization of the paper is as follows. Section \ref{sec:background} includes the notation used throughout the paper, and some background information on   measures of structuredness, conditional empirical entropy, and mixing processes. Section \ref{eq:from-lmep-to-qmap} establishes the  connection between Q-MAP and  L-MEP and how one can be derived from the other one.  Section \ref{sec:robustness} explores the robustness of Q-MAP to its estimate of the source distribution. Section \ref{sec:proofs} presents the proofs of the derived results and Section \ref{sec:conclusion} concludes the paper.


\section{Background}\label{sec:background}


\subsection{Notation}\label{sec:notation}
Given  $(x_1,x_2,\ldots,x_n)\in\mathds{R}^n$ and  $(i,j)\in\{1,\ldots,n\}^2$,  $x_i^j\triangleq (x_i,x_{i+1},\ldots,x_j).$ 
In the special cases when $i=1$ or $i=j$, $x_1^j$ and $x_j^j$ are denoted by $x^j$ and $x_j$, respectively.  Sets are denoted by calligraphic letters such as $\Xc$ and $\Yc$. The size of a finite set $\Xc$ is denoted by $|\Xc|$. For $a\in\mathds{R}$, $\delta_a$ denotes the Dirac measure on $\mathds{R}$ with an atom at $a$.  Throughout the paper, $\log$ and $\ln$ refer to logarithm in base 2 and the natural logarithm, respectively.

For $x\in\mathds{R}$, let $\lfloor x \rfloor$ denote the largest integer number smaller than or equal to $x$. For $b\in\mathds{N}^+$,   the $b$-bit quantized version of $x$ is denoted by $[x]_b$ and is defined as 
\begin{align}
[x]_b\triangleq\lfloor x\rfloor+\sum_{i=1}^b2^{-i}a_i,\end{align} 
where for all $i$, $a_i\in\{0,1\}$, and $0.a_1a_2\ldots$ denotes the binary representation of $x-\lfloor x\rfloor$.  Given $x^n\in\mathds{R}^n$, $[x^n]_b\triangleq([x_1]_b,\ldots,[x_n]_b).$ Given set $\Xc\subset\mathds{R}$, let $\Xc_b$ denote the b-bit quantized version of $\Xc$. That is,
\[
\Xc_b=\{[x]_b: \; x\in\Xc\}.
\]

\subsection{Measures of structuredness}\label{sec:measure-structure}

As described earlier, compressed sensing is about recovering structured real-valued signals from their linear under-sampled measurements. Hence, it is important to understand what it means for a real-valued signal to be structured and how one can measure the level of structuredness of a signal. One classic example of structured signals in compressed sensing is the set of  sparse signals.  Consider an independent identically distributed (i.i.d.) process $\Xbbf=\{X_i\}_{i=1}^{\infty}$, where $X_i\sim (1-p)\delta_0+pf_c$. Here, $p\in[0,1]$ and $f_c$ denotes the probability density function (pdf) of an absolutely continuous distribution. For any value $p\neq 1$, vector $X^n$, with high probability, is  sparse and hence  structured  and   can be recovered from its under-sampled measurements. However, for $p=1$, process $\Xbbf$ is an unstructured process and  $X^n$ cannot be recovered from less than $n$ measurements. For general stationary stochastic processes, it is desirable to have a measure of structuredness that automatically distinguishes  between the cases where compressed sensing is possible and the cases where it is not possible to recover the source from its undersampled measurements.  

Information dimension (ID) of a stochastic process  defined in \cite{JalaliP:17-IT} serves this purpose and provides a measure of structuredness for stationary stochastic processes. ID of a stationary process is closely related to its rate-distortion dimension, which is yet another measure of structuredness. (Refer to  \cite{kawabata1994rate}, \cite{geiger2017information} and \cite{RezagahJ:17} for more information on  this fundamental connection.)

\begin{definition}[ID of a stochastic process]
The $k$-th order upper  and lower IDs of a stochastic process $\Xbbf=\{X_i\}_{i=1}^{\infty}$ are denoted by $\bar{d}_k(\Xbbf)$ and $\underline{d}_k(\Xbbf)$ and  defined as  
\[
\bar{d}_k(\Xbbf)= \limsup_{b\to\infty} {H([X_{k+1}]_b|[X^k]_b)\over b},
\]
and 
\[
\underline{d}_k(\Xbbf)= \liminf_{b\to\infty} {H([X_{k+1}]_b|[X^k]_b)\over b},
\]
respectively. If $\lim_{k\to\infty}\bar{d}_k(\Xbbf)$ exists, the  upper  ID of process $\Xbbf$ is defined as $\bar{d}_o(\Xbbf)= \lim_{k\to\infty}\bar{d}_k(\Xbbf)$. Also, if $\lim_{k\to\infty}\underline{d}_k(\Xbbf)$ exists, the lower ID of process $\Xbbf$ is defined as $\underline{d}_o(\Xbbf)= \lim_{k\to\infty}\underline{d}_k(\Xbbf)$. In the case where $\bar{d}_o(\Xbbf)=\underline{d}_o(\Xbbf)$, the information dimension of process $\Xbbf$ is defined as $d_o(\Xbbf)=\bar{d}_o(\Xbbf)=\underline{d}_o(\Xbbf)$.
\end{definition}
It is straightforward to show that the ID of an i.i.d.~process $\Xbbf=\{X_i\}_{i=1}^{\infty}$ is equal to  the R\'enyi information dimension of its first order marginal distribution ($X_1$). In other words, the definition of ID for random processed is a generalization of R\'enyi's notion of ID, defined in \cite{renyi1959dimension}  for random variables and random vectors.   For the sparse i.i.d.~process described earlier, it can be proved that \cite{renyi1959dimension},
\[
d_o(\Xbbf)=p.
\]
This suggests that, at least in this special case, as desired, the ID of process $\Xbbf$  serves as a measure of structuredness. In this case, the maximum ID is achieved by processes with $p=1$, in which case, as explained earlier, compressed sensing is infeasible.  In fact, it can be proved that for all stationary processes $\Xbbf$ with $H(\lfloor X_1\rfloor)<\infty$, $\bar{d}_k(\Xbbf)\leq 1$ and  $\underline{d}_k(\Xbbf)\leq 1$ \cite{JalaliP:17-IT}. Moreover, for a  stochastic process $\Xbbf$ satisfying some  mixing conditions,  \cite{JalaliP:17-IT}  proves that  L-MEP  is a universal  algorithm that  recovers $X^n$  with almost zero distortion, from slightly more than $n\bar{d}_o(\Xbbf)$ random linear measurements. Therefore, intuitively speaking, in general, a stationary  process with an ID strictly smaller than one can be categorized as a structured process. On the other hand, processes whose ID is equal to one can be considered as unstructured processes that are not suitable for compressed sensing.

\subsection{Conditional empirical entropy}\label{sec:cond-emp-entropy}
In Section \ref{sec:measure-structure} we reviewed the definition of the ID of a stochastic process as a measure of its structuredness. In order to develop a universal compressed sensing algorithm, we also need a measure that  estimates the level of  structuredness of an individual   vector in $x^n\in\mathds{R}^n$. One approach to develop such a measure is to first quantize $x^n$ as $[x^n]_b$ to derive a discrete sequence and then  use one of the standard universal measures of complexity that already exist  in information theory, which are designed for sequences drawn from  discrete alphabets.  One such measure is the  conditional empirical entropy function that is also closely related to the definition of ID of a stochastic process is the  conditional  empirical entropy  \cite{cover}.

Consider sequence $u^n\in\Uc^n$, where $\Uc$ denotes a discrete set.  Given $k\in\mathds{N}^+$, the $k$-th order empirical distribution induced by $u^n$ is denoted by $\hat{p}_{k}(\cdot|u^n)$, and is defined as follows. For every $a^k\in\Uc^k$,
\begin{align}
  \hat{p}_{k}(a^k|u^n) & ={|\{i: u_{i-k}^{i-1}=a^k, k+1\leq i\leq n\}|\over n-k}\nonumber\\
&={1\over n-k}\sum_{i=k+1}^n\ind_{u_{i-k}^{i-1}=a^{k}},\label{eq:emp-dist}
\end{align}
where $\ind_{\Ec}$ denotes the indicator function of event $\Ec$. 

Conditional empirical entropy function is a well-known measure of complexity for  finite-alphabet sequences \cite{shields}.   The $k$-th order conditional empirical entropy of  $u^n$ is denoted by $\hat{H}_k(u^n)$ and is defined as
\[
\hat{H}_k(u^n)=H(U_{k+1}|U^k),
\]
where $U^{k+1}$ is distributed as $\hat{p}_{k+1}(\cdot|u^n)$. More explicitly,
\begin{align}
\hat{H}_k(u^n)&=-\sum_{a^{k+1}\in\Uc^{k+1}}\hat{p}_{k+1}(a^{k+1})\log {\hat{p}_{k+1}(a^{k+1})\over \hat{p}_{k}(a^{k})}.
\end{align}
Note that $  \hat{p}_{k}(a^{k})= \sum_{a_{k+1}\in\Uc} \hat{p}_{k+1}(a^{k+1})$.

\subsection{$\Psi^*$-mixing processes}
In this section, we briefly review $\Psi^*$-mixing processes. The reason we are interested in such processes is that for developing compressed sensing algorithms that are also applicable to sources with memory, we need to  be able to estimate the structuredness of different processes from their realizations. In Section \ref{sec:cond-emp-entropy}, we reviewed the conditional empirical entropy function which is one of the tools that enable us to develop such estimators. However, for such estimators to converge,  due to some technical challenges,  some additional constraints on the memory of the input  stochastic process is required. $\Psi^*$-mixing processes are a class processes which satisfy our desired convergence.   

Consider stochastic process $\Xbbf=\{X_i\}_{i=-\infty}^{\infty}$. Given $(j,k)\in\mathds{N}^2$, $j<k$, let $\Fc_j^k$ denote the smallest  $\sigma$ field containing all events generated by $X_j^k$. Then, the function $\psi^*_{\Xbbf}: \mathds{N}^+\to \mathds{R}^+$ is defined as
\begin{align*}
\psi^*_{\Xbbf}(g)\triangleq \sup_{(j,\Ac,\Bc)\in\mathds{N}\times \Fc_{-\infty}^j\times \Fc_{j+g}^{\infty}:  \P(\Ac)\P(\Bc)>0} {\P(\Ac\cap \Bc) \over \P(\Ac)\P(\Bc)},
\end{align*}
\begin{definition}[$\Psi^*$-mixing processes]
Stochastic process $\Xbbf$ is called $\Psi^*$-mixing, if 
\[
\lim_{g\to\infty}\psi^*_{\Xbbf}(g)=1.
\]
\end{definition}
Intuitively speaking, a stochastic process is $\Psi^*$-mixing, if its past and its future are almost independent from each other.   The following result  states the key property of $\Psi^*$-mixing processes that used to prove the main results about L-MEP and Q-MAP algorithms. 
\begin{theorem}[Theorem 6 in \cite{JalaliP:17-IT}]\label{thm:empirical-conv}
Consider a  stationary $\Psi^*$-mixing   process $\Xbbf$ and let process $\Zbbf$ denote the $b$-bit quantized version of process $\Xbbf$. That is, $Z_i=[X_i]_b$, for all $i$. Then,  for any $\epsilon>0$, there exists $g \in\mathds{ N}$, depending only on $\epsilon$ and the
function $\psi_{\Xbbf}(\cdot)$, such that for any $n > 6(k + g)/\epsilon + k$, 
\begin{align*}
\P&\LPr{\sum_{a^k\in\Xc_b^k}\left|\hat{p}_k(a^k|Z^n)-\P(Z^k=a^k)\right|\geq \epsilon}\nonumber\\
&\leq 
 2^{c\epsilon^2/8}(k + g)n^{|\Zc|^k}2^{-{nc\epsilon^2 \over 8(k+g) }},
\end{align*}
where $c=1/(2\ln 2)$.
\end{theorem}
\begin{remark}
 All memoryless processes and Markov sources with a discrete state space are $\Psi^*$-mixing \cite{shields}. Continuous space Markov processes are not in general $\Psi^*$-mixing.  However, the result of Theorem \ref{thm:empirical-conv}, \ie asymptotic convergence of  the empirical frequencies of their appropriately-quantized versions,   also holds for Markov processes that are weak $\Psi^*_q $-mixing,  defined in \cite{jalali2016qmap}.\footnote{For brevity, we skip stating the definition of this property and refer readers to \cite{jalali2016qmap}.} One example of such processes is a piecewise-constant signal modeled by a first-order Markov process.  All results in this paper mentioned for $\Psi^*$-mixing processes also hold for Markov processes that are weak $\Psi^*_q $-mixing.
 \end{remark}

\section{From L-MEP to Q-MAP}\label{eq:from-lmep-to-qmap}

Consider stationary process $\Xbbf =\{X_i\}_{i=1}^{\infty}$. Lagrangian-MEP (L-MEP) is a universal  algorithm that recovers $X^n$ from its measurements $Y^m=AX^n$, without knowing the source distribution. More precisely, given measurement $Y^m$ and sensing matrix $A$, L-MEP estimates $X^n$ as $\Xh_{\rm LMEP}^n$ by solving the following optimization:
 \begin{align}\label{eq:Lagrangian-MEP}
\Xh_{\rm LMEP}^n=\argmin_{x^n\in\Xc_b^n}\left[\hat{H}_k(u^n)+{\lambda\over n^2}\|Au^n-Y^m\|^2\right],
\end{align}
where $\lambda\in\mathds{R}^+$ (regularization coefficient), $b\in\mathds{N}^+$ (quantization level) and $k\in\mathds{N}^+$ (history size) are parameters that  need to be determined.  
The following theorem  proves that for the right choice of parameters, given enough measurements, L-MEP is a universal compressed sensing algorithm that recovers the source vector at almost zero distortion.
\begin{theorem}[Theorem 8 in \cite{JalaliP:17-IT}]\label{thm:LMEP}
Consider a stationary $\Psi^*$-mixing process $\Xbbf=\{X_i\}_{i=1}^{\infty}$ and let $Y^m=AX^n$, where $A\in\mathds{R}^{m\times n}$. The elements of $A$ are generated i.i.d.~$\Nc(0,1)$. Choose $\d>0$, and   $b=b_n=\lfloor r\log\log n\rfloor$, where $r>1$,
  $\l=\l_n=(\log n)^{2r}$ and
  \[
  m=m_n\geq (1+\d)\bar{d}_o(\Xbbf)n.
  \]
Also select a diverging  sequence $k=k_n$ such that  $k_n=o({\log n\over \log\log n})$.    Let $\Xh_{\rm LMEP}^n$ denote the solution of \eqref{eq:Lagrangian-MEP}.  Then, 
  \[
  {1\over \sqrt{n}}\|X^n-\Xh_{\rm LMEP}^n\|_2\stackrel{P}{\to}0.
  \]
\end{theorem}
Theorem \ref{thm:LMEP} implies that for $\Psi^*$-mixing processes, asymptotically, in the noiseless setting, as long as the sampling rate is larger than the upper ID of the source, L-MEP recovers the source almost with zero distortion, without knowing the source distribution or structure.

 L-MEP, described in \eqref{eq:Lagrangian-MEP}, employs the conditional empirical entropy function  as a measure of structuredness. However, conditional empirical entropy function is a highly non-linear function of  empirical distribution, and this makes finding the minimizer of \eqref{eq:Lagrangian-MEP} a challenging task. To address the computational complexity  issue and to move toward deriving  computationally-efficient algorithms, consider the following optimization, called approximate MEP (AMEP), where the conditional empirical entropy function is replaced by a linear function of the empirical distribution:
\begin{align}\label{eq:A-MEP}
\Xh^n_{\rm AMEP}=
&\argmin_{x^n\in\Xc_b^n}\left[\sum_{a^{k+1}\in\Xc_b^{k+1}} w_{a^{k+1}} \hat{p}_{k}(a^{k+1}|u^n) +{\lambda\over n^2}\|Au^n-Y^m\|^2\right],
\end{align}
Here, $\wb=(w_{a^{k+1}}: a^{k+1}\in\Xc_b^{k+1})$ are fixed non-negative weights that need to be determined.  As $\hat{H}_k(u^n)$ in L-MEP captures the structuredness of a candidate reconstruction sequence $u^n$,  $\sum_{a^{k+1}\in\Xc_b^{k+1}} w_{a^{k+1}} \hat{p}_{k}(a^{k+1}|u^n)$ is expected to play a similar role in  AMEP. This of course depends on the choice of the weights. Given weights $\wb=(w_{a^{k+1}}: a^{k+1}\in\Xc_b^{k+1})$, define function $c_{\wb}:\Xc_b^n\to \mathds{R}$, as follows. For $u^n\in\Xc_b^n$,
\begin{align}
c_{\wb}(u^n)=\sum_{a^{k+1}\in\Xc_b^{k+1}} w_{a^{k+1}} \hat{p}_{k}(a^{k+1}).
\end{align}
Using this definition, \eqref{eq:A-MEP} can be written as $\Xh^n_{\rm AMEP}=\argmin_{x^n\in\Xc_b^n}[c_{\wb}(u^n)+{\lambda\over n^2}\|Au^n-Y^m\|^2]$. 
To understand the performance of AMEP and the role of weights $\wb$ in capturing the level of structuredness of different sequences, first we need to determine a way to set the weights $\wb$.  After determining the weights, a key question is how the performances of  AMEP and LMEP compare with each other. Since  $\hat{H}_k$ is a highly non-linear function, it is not clear if there exists a set of weights that make the performance of the two optimizations  comparable. The following theorem  addresses both  questions and  proves that  \eqref{eq:Lagrangian-MEP} and \eqref{eq:A-MEP} in fact have the same performance, if the weights  $\wb$ are set appropriately.
\begin{theorem}\label{lemma:AMEP}
  Given measurements $y^m=Ax^n$ and measurement matrix $A\in\mathds{R}^{m\times n}$, let $\xh^n_{\rm LMEP}$  denote a minimizer of  L-MEP optimization described in \eqref{eq:Lagrangian-MEP}. Let $\wb=(w_{a^{k+1}}: a^{k+1}\in\Xc_b^{k+1})$ denote the set of partial derivatives of the $k$-order conditional empirical entropy function $ \hat{H}_k$ evaluated at $\hat{p}_{k+1}(\cdot|\xh^n_{\rm LMEP})$. That is, for $a^{k+1}\in\Xc_b^{k+1}$, let
  \begin{align}\label{eq:optimal-weights}
  w_{a^{k+1}}& = {\partial \hat{H}_k \over\partial \hat{p}_{k+1}(a^{k+1})}\Big|_{\hat{p}_{k+1}(\cdot|\xh^n_{\rm LMEP})}\nonumber\\
  &=\log {\hat{p}_k(a^k|\xh^n_{\rm LMEP})\over  \hat{p}_{k+1}(a^{k+1}|\xh^n_{\rm LMEP})}.
  \end{align}
Then the solution of the  AMEP optimization described in \eqref{eq:A-MEP} for the specified set of weights $\wb$ is  a minimizer of L-MEP optimization \eqref{eq:Lagrangian-MEP} as well.
\end{theorem}
The proof is presented in Section \ref{proof:AMEP}


Theorem \ref{lemma:AMEP} proves that   given the empirical distribution of a minimizer of \eqref{eq:Lagrangian-MEP},  we could  calculate the coefficients according to \eqref{eq:optimal-weights} and solve the optimization described in \eqref{eq:A-MEP} instead of the one described in  \eqref{eq:A-MEP}. The challenge of course is that it is not clear how to find the empirical distribution of the solution of \eqref{eq:Lagrangian-MEP}. In the noiseless setting,  L-MEP renders an almost zero-distortion reconstruction of the input sequence. In other words, the minimizer of \eqref{eq:Lagrangian-MEP}, asymptotically, is very close to the desired input signal $X^n$. Therefore, it seems plausible that if one computes the coefficients based on the empirical distribution of the quantized input signal instead of the empirical distribution of the of solution \eqref{eq:Lagrangian-MEP}, the performance of the derived algorithm is still close to that of L-MEP. The following theorem proves that this is in fact the case.

\begin{theorem}\label{thm:AMEP-input-signal}
Consider a stationary $\Psi^*$-mixing  process $\Xbbf $. Choose $r>1$, $\d>0$, and let $b=b_n=\lfloor r\log\log n\rfloor$,
  $\l=\l_n=(\log n)^{2r}$ and
  \[
  m=m_n\geq (1+\d)\bar{d}_o(\Xbbf)n.
  \]
Also, choose a diverging sequence $k=k_n$ such that $k_n=o({\log n\over \log\log n})$.    Assume that the decoder observes  measurements $Y^m=AX^n$, and the decoder's coefficients $\wb=(w_{a^{k+1}}: a^{k+1}\in\Xc_b^{k+1})$ are set as
  \begin{align}\label{eq:input-signal}
  w_{a^{k+1}} &= {\partial \hat{H}_k \over\partial \hat{p}_{k+1}(a^{k+1})}\Big|_{\hat{p}_{k+1}(\cdot|[X^n]_b)}\nonumber\\
  &=\log {\hat{p}_k(a^k|[X^n]_b)\over  \hat{p}_{k+1}(a^{k+1}|[X^n]_b)}.
  \end{align}
  The entries of the measurement matrix are generated i.i.d.~according to $\Nc(0,1)$ distribution. The decoder computes
  \[
\Xh^n=\argmin_{u^n\in\Xc_b^n}\left[c_{\wb}(u^n) +{\lambda\over n^2}\|Au^n-Y^m\|^2\right].
\]
Then, as $n$ grows to infinity,
\[
{1\over \sqrt{n}}\|X^n-\Xh^n\|_2\stackrel{\rm P}{\to} 0.
\]
\end{theorem}
The proof is presented in Section \ref{proof:AMEP-input}.

This result implies that if in addition to the measurements the decoder had access to the empirical distribution of the input signal,  it could recover the input signal asymptotically almost losslessly roughly from slightly more than $n\bar{d}_o(\Xbbf)$ random linear measurements using the simplified  optimization described \eqref{eq:A-MEP}.

Theorems \ref{lemma:AMEP} and \ref{thm:AMEP-input-signal} provide  two different ways to set the coefficients $\wb$ in the A-MEP optimization, described in \eqref{eq:A-MEP}, so that asymptotically it has the same performance as the L-MEP optimization. The prescribed   set of weights  are   derived either  from  the empirical distribution of the quantized  input sequence $X^n$, or from  the empirical distribution of  the quantized version of its reconstruction $\Xh^n_{\rm LMEP}$.  Therefore, the resulting algorithms for the specified set of weights  are not universal compressed sensing algorithms. However, for $\Psi^*$-mixing  stationary processes, using Theorem \ref{thm:empirical-conv}, the empirical probabilities  of the quantized  input sequence $X^n$ is expected to be  close to their expected values, given that $k$ and $b$ grow slowing enough with the blocklength $n$.  This raises the following question: 
\begin{question}
Can we replace the input-dependent weights with a set of weights that only depend on the source distribution and still have the same 
performance guarantee? 
\end{question}
The answer to the above question is in fact affirmative. Deriving the weights using the input distribution  instead of using the $(k+1)$ order empirical distribution of $[X^n]_b$ or  $[\Xh^n_{\rm LMEP}]_b$ yields  the Q-MAP optimization, which is a Bayesian compressed sensing recovery algorithm. More precisely, given stationary process $\Xbbf$, for $a^{k+1}\in\Xc_b^{k+1}$, let 
\begin{align}
w_{a^{k+1}}=-\log \P\left([X_{k+1}]_b=a_{k+1}|[X^k]_b=a^k\right).\label{eq:qmap-weights}
\end{align}
Note that this new set of  weights are not random variables and do not depend on any empirical distribution. The following theorem proves that the  Q-MAP  optimization,  for any fixed $k$, under some mild technical condition, asymptotically,  is still able to recover the input vector $X^n$ from measurements $Y^m=AX^n$,  at any sampling rate slightly larger than $d_k(\Xbbf)$. 
\begin{theorem}[Theorem 5.1 in \cite{jalali2016qmap}]\label{thm:mainresult-Q-MAP-L}
Consider  a $\Psi^*$-mixing stationery  process $\Xbbf$.  Let $Y^m=AX^n$, where the entries of $A$ are i.i.d. $\Nc(0,1)$. Choose  $k$, $r>1$ and $\d>0$, and let $b=b_n=\lceil r\log\log n\rceil$,  $\l=\l_n=(\log n)^{2r}$ and  $m=m_n\geq (1+\delta)n\bar{d}_k(\Xbbf)$.  Assume that  there exists a constant  $f_{k+1}>0$, such that  for  any quantization level $b$, and any  $u^{k+1}\in\Xc_b^{k+1}$ with  $\P([X^{k+1}]_b=u^{k+1})\neq 0$, 
\begin{align}
\P([X^{k+1}]_b=u^{k+1})\geq {f_{k+1} |\Xc_b|^{-(k+1)}}.\label{eq:cond-f-k}
\end{align}
 Further, assume that $\Xh^n_{\rm QMAP}$ denotes the solution of  \eqref{eq:A-MEP}, where the coefficients are computed according to \eqref{eq:qmap-weights}. Then, for any $\e>0$,
 \begin{align*}
\lim_{n\to\infty} \P\left( { 1\over \sqrt{n}}\|X^n-\Xh^n_{\rm QMAP}\|_2>\e\right)=0.
  \end{align*}
\end{theorem} 
\begin{remark}
As described in \cite{jalali2016qmap}, the condition stated in \eqref{eq:cond-f-k}    is a technical condition that is needed in the proof of Theorem  \ref{thm:mainresult-Q-MAP-L}. In \cite{jalali2016qmap}, the authors conjecture that the result holds even for distributions that does  not satisfy this condition. The results of Lemma   \ref{lemma:AMEP} and Theorem \ref{thm:AMEP-input-signal} suggest that this might in fact  be the case .\end{remark}

\section{Robustness to coefficients}\label{sec:robustness}
In the previous section, we characterized three sets  of weights that  enable the A-MEP optimization described in \eqref{eq:A-MEP} to recover the input sequence asymptotically at zero-distortion:
\begin{enumerate}
\item In the first case,  for $a^{k+1}\in\Xc_b^{k+1}$,   we set $w_{a^{k+1}}=\log {\hat{p}_k(a^k|\xh^n_{\rm LMEP})\over  \hat{p}_{k+1}(a^{k+1}|\xh^n_{\rm LMEP})}$, where $\hat{p}_{k+1}(a^{k+1}|\xh^n_{\rm LMEP})$ denotes the $(k+1)$-th order empirical distribution of $\xh^n_{\rm LMEP}$. Theorem \ref{lemma:AMEP} evaluates the performance of the resulting algorithm. 
\item In the second case,   for $a^{k+1}\in\Xc_b^{k+1}$,   we set $w_{a^{k+1}}=\log {\hat{p}_k(a^k|[X^n]_b)\over  \hat{p}_{k+1}(a^{k+1}|[X^n]_b)}.$ In other words, in this case the weights are set based on the $(k+1)$-th order empirical distribution of the quantized input sequence.  Theorem \ref{thm:AMEP-input-signal} evaluates the asymptotic  performance of the resulting algorithm. 
\item Finally, in the third  case,  for $a^{k+1}\in\Xc_b^{k+1}$,   we set $w_{a^{k+1}}=-\log \P\left([X_{k+1}]_b=a_{k+1}|[X^k]_b=a^k\right)$. That is, in this case the weights are computed based on the known distribution of the input. The resulting algorithm is the Q-MAP algorithm whose performance is stated in Theorem \ref{thm:mainresult-Q-MAP-L}.
\end{enumerate}
In summary, based on these three recipes,  to evaluate the weights, one needs to have access to one of the following: the solution of the L-MEP optimization, the quantized input sequence, or the source  distribution. While none of these three options, especially the first two, seem to be very  practical, they inspire the following alternative path for deriving the weights: learn the source distribution from available training data and use it to evaluate the weights. 
This raises the following important questions about the robustness of the A-MEP algorithm to its inputs weights.
\begin{question}\label{Q:2}
How sensitive is the performance of  A-MEP (or Q-MAP) optimization to the set of weights $\{w_{a^{k+1}}: \;a^{k+1}\in\Xc_b^{k+1}\}$? In other words, could we disturb the weights and still, asymptotically,  recover the source  at zero-distortion at the same sampling rate?
\end{question}
\begin{question}\label{Q:3}
Assume that  there is a mismatch between the  input distribution and the distribution used to  derive the weights. How much should  the sampling rate  be increased to compensate for the mismatch?
\end{question}
In this section we answer both of these questions.

Consider a stochastic  process $\Xbbf =(X_i: i=1,2,\ldots)$. For source signal $X^n$, let $\Wb^*=\Wb^*_n(X^n)$ denote the  coefficients  calculated based on  the empirical distribution of the quantized version of the input signal $X^n$. That is, for $a^{k+1}\in\Xc_b^{k+1}$,
\begin{align}
W^*_{a^{k+1}}=- \log \hat{p}_{k+1}(a_{k+1}|a^k),\label{eq:def-W*}
\end{align}
where $\hat{p}_{k+1}=\hat{p}_{k+1}(\cdot|[X^n]_b)$.
The following theorem addresses Question \ref{Q:2} and proves that if instead of $\Wb^*$, the A-MEP decoder  employes  a weight vector $\Wb$, which with high probability is close to $\Wb^*$ , then its asymptotic recovery  performance is not affected. More precisely, using the same sampling rate,  with high probability, A-MEP   is still able to recover the source  from noise-free measurements $Y^m=AX^n$, almost losslessly. 

\begin{theorem}\label{thm:noisycoeffsbound}
 Consider the setup of Theorem \ref{thm:AMEP-input-signal}. Assume that the vector of the coefficients  ${\Wb}={\Wb}_n$ is such that, for any $\e>0$,
  \[
  \P\Big({1\over b}\|{\Wb}_n-\Wb^*_n\|_{\infty}>\e \Big)\to 0,
  \]
  as $n$ grows without bound. Let $\Xh^n= \argmin_{u^n\in\Xc_b^n}[c_{\Wb}(u^n) +{\lambda\over n^2}\|Au^n-Y^m\|^2]$. 
  Then, as $n$ grows to infinity, 
  \[
  {1\over \sqrt{n}}\|X^n-\Xh^n\|_2\stackrel{\P}{\to}0.
  \]
\end{theorem}
The proof is presented in Section \ref{app:A}.

Theorem \ref{thm:noisycoeffsbound} proves that if the distance between the weights used by the A-MEP algorithm, described in \eqref{eq:A-MEP}, and the weights derived from $[X^n]_b$ is small enough, the performance of A-MEP, asymptotically, does not change. In practice, typically,   the distribution of the source is learned from some  available datasets. Therefore, it is important to also connect the performance of the algorithm to the distance between the learned distribution and the true underlying distribution.  The next theorem directly looks  at the distance between the original distribution of the source and its learned version and shows that if this distance is small enough the asymptotic performance does not change.

\begin{theorem}\label{thm:noisycoeffsbound-Dkl}
 Consider the setup of Theorem \ref{thm:AMEP-input-signal}. Assume that the vector of the coefficients  ${\Wb}={\Wb}_n$ is such that for every $n$, there exists a probability distribution $q_{k+1}$ over $\Xc_b^{k+1}$ such that for every   $a^{k+1}\in\Xc_b^k$,
 \[
W_{a^{k+1}}=- \log q_{k+1}(a_{k+1}|a^k),
\]
and 
  \[
  \P\Big({1\over b}D_{\rm KL}(\hat{p}_{k+1}(\cdot| [X^n]_b),q_{k+1}))>\e \Big)\to 0,
  \]
  as $n$ grows without bound. Let  $\Xh^n= \argmin_{u^n\in\Xc_b^n}[c_{\Wb}(u^n) +{\lambda\over n^2}\|Au^n-Y^m\|^2].$
  Then,
  \[
  {1\over \sqrt{n}}\|X^n-\Xh^n\|_2\stackrel{\P}{\to}0.
  \]

\end{theorem}
The proof is presented in Section \ref{app:B}.

We next turn to Question \ref{Q:3} and show that if there is a bounded mismatch between the weights $\Wb^*$ and the ones used by the A-MEP algorithm, $\Wb$, by increasing the sampling rate with a constant proportional to the $\ell_{\infty}$ distance between $\Wb$ and $\Wb^*$,  asymptotically, A-MEP still recovers the source vector  almost losslessly. 

\begin{theorem}\label{cor:nodel-mismatch}
 Consider the setup of Theorem \ref{thm:AMEP-input-signal}. Assume that the vector of the coefficients  ${\Wb}={\Wb}_n$ is such that,   as $n$ grows without bound,
   \[
  \P\Big({1\over b}\|{\Wb}_n-\Wb^*_n\|_{\infty}>\e_w \Big)\to 0,
  \]
for some $\e_w>0$. Assume that
   \[
  m=m_n\geq (1+\d)(\bar{d}_o(\Xbbf)+3\e_w)n,
  \]
  where $\d>0$ is a free parameter. Let $\Xh^n= \argmin_{u^n\in\Xc_b^n}[c_{\Wb}(u^n) +{\lambda\over n^2}\|Au^n-Y^m\|^2]$.
  Then,
  \[
  {1\over \sqrt{n}}\|X^n-\Xh^n\|_2\stackrel{\P}{\to}0.
  \]
\end{theorem}
The proof is presented in Section \ref{app:C}.

\section{Proofs}\label{sec:proofs}
In this section we provide the proofs of the results mentioned earlier in the paper. 

\subsection{Proof of Theorem \ref{proof:AMEP-input}}\label{proof:AMEP-input}
Since $\Xh^n$ is a minimizer of $c_{\wb}(u^n)+{\lambda\over n^2}\|Au^n-Y^m\|^2$ over all $u^n\in\Xc_b^n$, we have
\begin{align}\label{eq:cost-comapre-1}
  &\sum_{a^{k+1}\in\Xc_b^{k+1}} w_{a^{k+1}} \hat{p}_{k}(a^{k+1}|\Xh^n) +{\lambda\over n^2}\|A\Xh^n-Y^m\|^2\nonumber\\
   &\;\;\leq  \sum_{a^{k+1}\in\Xc_b^{k+1}} w_{a^{k+1}} \hat{p}_{k}(a^{k+1}|[X^n]_b) +{\lambda\over n^2}\|A[X^n]_b-Y^m\|^2\nonumber \\
   &\;\;=\hat{H}_k([X^n]_b)+ {\lambda\over n^2}\|A[X^n]_b-Y^m\|^2,
\end{align}
where the last line follows because, using the weights from \eqref{eq:input-signal}, we have 
\begin{align*}
c_{\wb}(u^n)&= \sum_{a^{k+1}\in\Xc_b^{k+1}} w_{a^{k+1}} \hat{p}_{k}(a^{k+1}|[X^n]_b)  \nonumber\\
 &= \sum_{a^{k+1}\in\Xc_b^{k+1}} \hat{p}_{k}(a^{k+1}|[X^n]_b)\log {\hat{p}_k(a^k|[X^n]_b)\over  \hat{p}_{k+1}(a^{k+1}|[X^n]_b)} \\
 &=\hat{H}_k([X^n]_b).
\end{align*}
On the other hand, by the concavity of the empirical entropy function \cite{JalaliM:12}, it follows that
\begin{align}\label{eq:concavity-lemma2}
  \hat{H}_k(\Xh^n) &\leq \hat{H}_k([X^n]_b)+ \sum_{a^{k+1}\in\Xc_b^{k+1}} w_{a^{k+1}} \Big(\hat{p}_{k}(a^{k+1}|\Xh^n)-\hat{p}_{k}(a^{k+1}|[X^n]_b)\Big)\nonumber\\
  &=  \sum_{a^{k+1}\in\Xc_b^{k+1}} w_{a^{k+1}} \hat{p}_{k}(a^{k+1}|\Xh^n).
\end{align}
Adding ${\lambda\over n^2}\|A\Xh^n-Y^m\|^2$ to the both sides of \eqref{eq:concavity-lemma2}, we derive
\begin{align}\label{eq:cost-compare-2}
    \hat{H}_k&(\Xh^n)+{\lambda\over n^2}\|A\Xh^n-Y^m\|^2 \nonumber\\
     &\leq  \sum_{a^{k+1}\in\Xc_b^{k+1}} w_{a^{k+1}} \hat{p}_{k}(a^{k+1}|\Xh^n)+{\lambda\over n^2}\|A\Xh^n-Y^m\|^2.
\end{align}
Therefore, combining \eqref{eq:cost-comapre-1} and \eqref{eq:cost-compare-2}, we have
\begin{align}
    \hat{H}_k(\Xh^n)+{\lambda\over n^2}\|A\Xh^n-Y^m\|^2  &\leq   \hat{H}_k([X^n]_b)+{\lambda\over n^2}\|A[X^n]_b-Y^m\|^2
\end{align}
The rest of the proof follows from the proof of  Theorem 8 in \cite{JalaliP:17-IT}. The reason is that while L-MEP seeks the minimizer of a different cost function ($\hat{H}_k(u^n)+{\lambda\over n^2}\|Au^n-Y^m\|^2$), the proof only uses the fact that the cost of the output $\Xh^n$ is smaller than the cost of the quantized version of the input ($[X^n]_b$).

\subsection{Proof of Theorem \ref{lemma:AMEP}}\label{proof:AMEP}
  For the ease of the notation, let $u^n=\xh^n_{\rm LMEP}$ and $v^n=\xh^n_{\rm AMEP}$. Since $\hat{H}_k$ is a concave function of $\hat{p}_{k+1}$ \cite{JalaliM:12}, and since by assumption the weights $\wb=(w_{a^{k+1}}: a^{k+1}\in\Xc_b^{k+1})$ are computed at $\hat{p}_{k+1}(u^n)$, we have
  \begin{align}\label{eq:concavity-Hk-1}
    \hat{H}_k(v^n) & \leq \hat{H}_k(u^n) + \sum_{a^{k+1}}w_{a^{k+1}} (\hat{p}_{k+1}(a^{k+1}|v^n)-\hat{p}_{k+1}(a^{k+1}|u^n))\nonumber\\
    &=\sum_{a^{k+1}}w_{a^{k+1}} \hat{p}_{k+1}(a^{k+1}|v^n),
  \end{align}
  where the last line follows because    
  \begin{align*}
\sum_{a^{k+1}\in\Xc_b^{k+1}}&w_{a^{k+1}} \hat{p}_{k+1}(a^{k+1}|u^n)\nonumber\\
&=\sum_{a^{k+1}\in\Xc_b^{k+1}} \hat{p}_{k+1}(a^{k+1}|u^n) \log {\hat{p}_k(a^k|u^n)\over  \hat{p}_{k+1}(a^{k+1}|u^n)}\nonumber\\
&=\hat{H}_k(u^n).
  \end{align*}
Adding ${\lambda\over n^2}\|Av^n-y^m\|^2$ to the both sides of \eqref{eq:concavity-Hk-1}, it follows that
   \begin{align}\label{eq:concavity-Hk-2}
    &\hat{H}_k(v^n)+{\lambda\over n^2}\|Av^n-Y^m\|^2 \nonumber\\
    & \leq
     \sum_{a^{k+1}}w_{a^{k+1}} \hat{p}_{k+1}(a^{k+1}|v^n)+{\lambda\over n^2}\|Av^n-y^m\|^2.
  \end{align}
  On the other hand, by our assumption, $v^n$ is a minimizer of \eqref{eq:A-MEP}. Therefore,
  \begin{align}\label{eq:minimizer-AMEP}
   \sum_{a^{k+1}\in\Xc_b^{k+1}}& w_{a^{k+1}} \hat{p}_{k+1}(a^{k+1}|v^n) +{\lambda\over n^2}\|Av^n-y^m\|^2 \nonumber\\
    &\leq  \sum_{a^{k+1}\in\Xc_b^{k+1}} w_{a^{k+1}} \hat{p}_{k+1}(a^{k+1}|u^n) +{\lambda\over n^2}\|Au^n-y^m\|^2 \nonumber\\
    &= \hat{H}_k(u^n)+{\lambda\over n^2}\|Au^n-y^m\|^2.
  \end{align}
  Therefore, combining \eqref{eq:concavity-Hk-2} and \eqref{eq:minimizer-AMEP}, we have
  \[
    \hat{H}_k(v^n)+{\lambda\over n^2}\|Av^n-y^m\|^2\leq
     \hat{H}_k(u^n)+{\lambda\over n^2}\|Au^n-y^m\|^2.
  \]
  But  $u^n$ is a minimizer of \eqref{eq:Lagrangian-MEP}, which implies
  \[
    \hat{H}_k(u^n)+{\lambda\over n^2}\|Au^n-y^m\|^2\leq
     \hat{H}_k(v^n)+{\lambda\over n^2}\|Av^n-y^m\|^2.
  \]
  Therefore, $  \hat{H}_k(u^n)+{\lambda\over n^2}\|Au^n-y^m\|^2=
     \hat{H}_k(v^n)+{\lambda\over n^2}\|Av^n-Y^m\|^2$, and as a result $u^n$ is also a minimizer of \eqref{eq:Lagrangian-MEP}.
\subsection{Proof of  Theorem \ref{thm:noisycoeffsbound}}\label{app:A}
Before stating the proof of Theorem \ref{thm:noisycoeffsbound}, we prove the following  lemma which is used in the proof. 

\begin{lemma}\label{lemma:w-what-cost-diff} 
Given $\wb=(w_{a^{k+1}}: a^{k+1}\in\Xc_b^{k+1})$ and $\hat{\wb}=(\hat{w}_{a^{k+1}}:a^{k+1}\in\Xc_b^{k+1})$, $A\in\mathds{R}^{m\times n}$ and $y^m\in\mathds{R}^m$, define functions $f:\Xc_b^n\to\mathds{R}$ and $\hat{f}:\Xc_b^n\to\mathds{R}$, as
\[
f(u^n)=\sum_{a^{k+1}\in\Xc_b^{k+1}} w_{a^{k+1}} \hat{p}_{k}(a^{k+1}|u^n) +{\lambda\over n^2}\|Au^n-Y^m\|^2\
 \]
and 
\[
\hat{f}(u^n)=\sum_{a^{k+1}\in\Xc_b^{k+1}} \hat{w}_{a^{k+1}} \hat{p}_{k}(a^{k+1}|u^n) +{\lambda\over n^2}\|Au^n-Y^m\|^2,
\]
respectively. Assume that the weights $\wb$ and $\hat{\wb}$ are such that 
\[
\|\wb-\hat{\wb}\|_{\infty}\leq \e,
\]
for some $\e>0$.   Let $\Xh^n\triangleq \argmin f(u^n)$ and $\Xt^n\triangleq \argmin \hat{f}(u^n)$. Then,
  \[
  f(\Xh^n)\leq f(\Xt^n) \leq f(\Xh^n)+2\e.
  \]
\end{lemma}

\begin{proof}
Since by definition $\Xh^n$ and $\Xt^n$ are the minimizers of $f$ and $\hat{f}$, respectively, we have
\[
f(\Xh^n)\leq f(\Xt^n),
\]
and
\[
\hat{f}(\Xt^n)\leq \hat{f}(\Xh^n).
\]
On other hand, since $\|\wb-\hat{\wb}\|_{\infty}\leq \e$, we have
\begin{align}\label{eq:f-to-fhat}
 f(\Xt^n) &= \sum_{a^{k+1}\in\Xc_b^{k+1}} w_{a^{k+1}} \hat{p}_{k}(a^{k+1}|\Xt^n) +{\lambda\over n^2}\|A\Xt^n-Y^m\|^2\nonumber\\
 &\leq   \sum_{a^{k+1}\in\Xc_b^{k+1}} (\hat{w}_{a^{k+1}}+\e) \hat{p}_{k}(a^{k+1}|\Xt^n) +{\lambda\over n^2}\|A\Xt^n-Y^m\|^2\nonumber\\
 &= \hat{f}(\Xt^n)+ \e \sum_{a^{k+1}\in\Xc_b^{k+1}} \hat{p}_{k}(a^{k+1}|\Xt^n) \nonumber\\
 &= \hat{f}(\Xt^n)+ \e\nonumber\\
 &\leq \hat{f}(\Xh^n)+\e.
\end{align}
Similarly,
\begin{align}\label{eq:fhat-to-f}
 \hat{f}(\Xh^n) &= \sum_{a^{k+1}\in\Xc_b^{k+1}} \hat{w}_{a^{k+1}} \hat{p}_{k}(a^{k+1}|\Xh^n) +{\lambda\over n^2}\|A\Xh^n-Y^m\|^2\nonumber\\
 &\leq   \sum_{a^{k+1}\in\Xc_b^{k+1}} ({w}_{a^{k+1}}+\e) \hat{p}_{k}(a^{k+1}|\Xh^n) +{\lambda\over n^2}\|A\Xh^n-Y^m\|^2\nonumber\\
 &= {f}(\Xh^n)+ \e \sum_{a^{k+1}\in\Xc_b^{k+1}} \hat{p}_{k}(a^{k+1}|\Xh^n) \nonumber\\
 &= {f}(\Xh^n)+ \e.
\end{align}
Therefore, combining \eqref{eq:f-to-fhat} and \eqref{eq:fhat-to-f}, it follows that
\[
f(\Xh^n)\leq f(\Xt^n) \leq f(\Xh^n)+2\e.
\]
\end{proof}
\begin{proof}[Proof of  Theorem \ref{thm:noisycoeffsbound}]
Define event $\Ec_1$ as
\[
\Ec_1=\{{1\over b}\|{\Wb}_n-\Wb^*_n\|_{\infty}\leq \e\}.
\]
Let
\[
\Xt^n=  \argmin_{u^n\in\Xc_b^n}\Big[\sum_{a^{k+1}\in\Xc_b^{k+1}} W^*_{a^{k+1}} \hat{p}_{k}(a^{k+1}|u^n) +{\lambda\over n^2}\|Au^n-Y^m\|^2\Big].
\]
 Conditioned on $\Ec_1$, we have
\begin{align}\label{eq:f-what-vs-Xn-b}
  \sum_{a^{k+1}}&W_{a^{k+1}}\hat{p}_{k+1}(a^{k+1}|\Xh^n) +{\lambda\over n^2}\|Y^m-A\Xh^n\|^2 \nonumber\\
   &\stackrel{(a)}{\leq} 2\e b+ \sum_{a^{k+1}}W^*_{a^{k+1}}\hat{p}_{k+1}(a^{k+1}|\Xt^n) +{\lambda\over n^2}\|Y^m-A\Xt^n\|^2\nonumber\\
   &\stackrel{\rm (b)}{\leq} 2\e b+ \sum_{a^{k+1}}W^*_{a^{k+1}}\hat{p}_{k+1}(a^{k+1}|[X^n]_b) +{\lambda\over n^2}\|Y^m-A[X^n]_b\|^2\nonumber\\
   &= 2\e b + \hat{H}_k([X^n]_b)+ {\lambda\over n^2}\|Y^m-A[X^n]_b\|^2\nonumber\\
   &=  2\e b + \hat{H}_k([X^n]_b)+ {\lambda\over n^2}\|A(X^n-[X^n]_b)\|^2\nonumber\\
   &\stackrel{\rm (c)}{\leq}  2\e b + \hat{H}_k([X^n]_b)+ {\lambda (\sigma_{\max}(A))^2 \over n^2} \|X^n-[X^n]_b\|^2,
\end{align}
where (a) follows from Lemma \ref{lemma:w-what-cost-diff}, (b) holds because $\Xt^n$, by assumption,  is a minimizer of  $\sum_{a^{k+1}\in\Xc_b^{k+1}}W^*_{a^{k+1}} \hat{p}_{k}(a^{k+1}|u^n) +{\lambda\over n^2}\|Au^n-Y^m\|^2$, and (c) holds because $\|A u^n\|\leq \sigma_{\max}(A)\|u^n\|$, for all $u^n\in\mathds{R}^n$. But,
\[
\|X^n-[X^n]_b\|\leq 2^{-b}\sqrt{n}.
\]
Define events $\Ec_2$ and $\Ec_3$ as
\[
\Ec_2=\{\sigma_{\max}(A)<\sqrt{n}+2\sqrt{m}\},
\]
and
\[
\Ec_3=\{{1\over b} \hat{H}_k([X^n]_b)\leq \bar{d}_o(\Xbbf)+\d\},
\]
where $\d>0$. As proved in \cite{CaTa05},
\[
\P(\Ec_2^c)\leq 2^{-m/2}.
\]
Also,  given our choice of parameters $k$ and $b$, for a $\Psi^*$-mixing  process, $\P(\Ec_3^c)$ converges to zero, as $n$ goes to infinity \cite{JalaliP:17-IT}.

On the other hand, conditioned on $\Ec_1\cap \Ec_2\cap \Ec_3$, it follows from \eqref{eq:f-what-vs-Xn-b} that
 \begin{align}\label{eq:f-what-vs-Xn-b-result}
    {1\over b}\sum_{a^{k+1}}&W_{a^{k+1}}\hat{p}_{k+1}(a^{k+1}|\Xh^n) +{\lambda\over n^2b}\|Y^m-A\Xh^n\|^2 \nonumber\\
    &\leq \bar{d}_o(\Xbbf)+\d + 2\e+ {\l(\sqrt{n}+2\sqrt{m})^22^{-2b}\over nb}.
 \end{align}
 But, since $\l=(\log n)^{2r}$ and $b=\lceil r\log\log n\rceil$, we have
 \begin{align*}
 {\l(\sqrt{n}+2\sqrt{m})^22^{-2b}\over nb}&\leq {(\log n)^{2r}(\sqrt{n}+2\sqrt{m})^2\over nb2^{-2r\log\log n}}\\
 &= {(1+2\sqrt{m/n})^2\over b}\\
 &\leq {6\over b}.
 \end{align*}
 Therefore, conditioned on $\Ec_1\cap \Ec_2\cap \Ec_3$, since both terms on the left hand side of \eqref{eq:f-what-vs-Xn-b-result} are positive, we have
\begin{align}\label{eq:cost-Xh-W-vs-do}
{1\over b} \sum_{a^{k+1}}W_{a^{k+1}}\hat{p}_{k+1}(a^{k+1}|\Xh^n) & \leq \bar{d}_o(\Xbbf)+\e',
\end{align}
and
\begin{align}\label{eq:cost-Xh-dist-vs-do}
{\l\over n^2 b} \| Y^m-A\Xh^n\|^2 & \leq \bar{d}_o(\Xbbf)+\e',
\end{align}
where
\[
 \e'\triangleq 2\e +\d+{6\over b},
\]
can be made arbitrary small.

Conditioned on $\Ec_1$, $\|\Wb-\Wb^*\|_{\infty}\leq b\e$. Therefore, conditioned on $\Ec_1$,
\begin{align}\label{eq:min-cost-W-cost-W-hat}
  {1\over b} \sum_{a^{k+1}}W_{a^{k+1}}\hat{p}_{k+1}(a^{k+1}|\Xh^n)  & \geq {1\over b} \sum_{a^{k+1}}(W^*_{a^{k+1}}-\e b)\hat{p}_{k+1}(a^{k+1}|\Xh^n)\nonumber\\
  &= {1\over b} \sum_{a^{k+1}}W^*_{a^{k+1}}\hat{p}_{k+1}(a^{k+1}|\Xh^n)-\e.
\end{align}

 Define  $q_1$ and $q_2$ as the $(k+1)$-th order empirical distributions induced by $[X^n]_b$ and $\Xh^n$, respectively. That is,
 \[q_1=\hat{p}_{k+1}(\cdot|[X^n]_b)
 \]
 and
 \[
 q_2=\hat{p}_{k+1}(\cdot|\Xh^n).
 \]
 Then,
\begin{align}\label{eq:costW-Xh-KL}
 \sum_{a^{k+1}}&W^*_{a^{k+1}}\hat{p}_{k+1}(a^{k+1}|\Xh^n) \nonumber\\
  &=   \sum_{a^{k+1}}W^*_{a^{k+1}}q_2(a^{k+1})\nonumber
      \end{align}
                    \begin{align}
 &=  \sum_{a^{k+1}}q_2(a^{k+1})\log{1\over q_1(a_{k+1}|a^k)}\nonumber\\
 &= \sum_{a^{k+1}}q_2(a^{k+1})\log{q_2(a_{k+1}|a^k)\over q_1(a_{k+1}|a^k)}+
 \sum_{a^{k+1}}q_2(a^{k+1})\log{1\over q_2(a_{k+1}|a^k)} \nonumber\\
 &= \sum_{a^k}q_2(a^k) D_{\rm KL}(q_2(\cdot|a^k)\|q_1(\cdot|a^k)) +
 \hat{H}_{k}(\Xh^n).
\end{align}
Since $\sum_{a^k}q_2(a^k) D_{\rm KL}(q_2(\cdot|a^k)\|q_1(\cdot|a^k))\geq 0$, combining \eqref{eq:cost-Xh-W-vs-do}, \eqref{eq:min-cost-W-cost-W-hat} and \eqref{eq:costW-Xh-KL} yields
\begin{align}\label{eq:final-ub-Hhat-Xh-n}
   {1\over b}\hat{H}_{k}(\Xh^n) & \leq \bar{d}_o(\Xbbf)+\e'+\e.
\end{align}
Note that $\e$ and $\e'$ can be made arbitrary small, and $\P(\Ec_1^c\cup\Ec_2^c\cup\Ec_3^c)$ goes to zero as $n$ grows to infinity. The rest of the proof follows similar to the final steps of the  proof of Theorem 5 in \cite{JalaliP:17-IT}. We include a summary of the remaining steps for completeness. Define set $\Cc_n$ as
\[
 \Cc_n=\{u^n\in\Xc_b^n: {1\over nb}\ell_{\rm LZ}(u^n)\leq \bar{d}_o(\Xbbf)+4\e\}.
\]
Here $\ell_{\rm LZ}(u^n)$ denotes the length of the binary encoding of $u^n$ using the Lempel-Ziv encoder \cite{LZ}. 
Then, choosing $\d$ small enough, for our choice parameters $k=k_n$ and $b=b_n$, conditioned on $\Ec_1\cap\Ec_2\cap \Ec_3$, for all $n$ large enough, $[X^n]_b\in\Cc_n$ and $[\Xh^n]_b\in\Cc_n$.\footnote{For details on the connection between $\hat{H}_k$ and $\ell_{\rm LZ}$, refer to Appendix A in \cite{JalaliP:17-IT}.} But, since $\rm LZ$ is a uniquely decodable code,
\[
|\Cc_n|\leq 2^{nb(\bar{d}_o(\Xbbf)+4\e)}.
\] 
Now define event $\Ec_4$ as
\[
\Ec_4=\{\|A(u^n-X^n)\|\geq \|u^n-X^n\|\sqrt{(1-\tau)m}: \forall u^n\in\Cc_n\},
\]
where $\tau>0$ is a free parameter. For  fixed $u^n\in\Cc_n$ and fixed $X^n=x^n$, using a lemma  in \cite{JalaliM:14-MEP-IT} on the concentration of $\chi^2$ random variables, we have
\[
       \P\big(A(u^n-x^n)\|\leq \|u^n-x^n\|\sqrt{(1-\tau)m}\Big)\leq\mathrm{e}^{\frac{m}{2}(\tau+\ln(1-\tau))}
   \]
Therefore, by the union bound,   for a fixed vector $X^n$,
\[
\P_A(\Ec_3^c)\leq  2^{nb(\bar{d}_o(\Xbbf)+4\e)} \ex^{{m\over 2}(\tau +\ln (1-\tau))},
\]
Here $\P_A$ reflects the fact that $[X^n]_b$ is fixed, and the randomness is in the generation of matrix $A$. Proper choice of parameter $\tau$ combined with the Fubini's Theorem and the Borel Cantelli Lemma proves that $\P_{X^n}(\Ec_3^c)\to 0$, almost surely. Conditioned on $\Ec_1\cap\Ec_2\cap\Ec_3\cap\Ec_4$,  $\Xh^n\in\Cc_n$, and therefore from \eqref{eq:cost-Xh-dist-vs-do},
\begin{align}\label{eq:bound-error-triangle}
 {\lambda(1-\tau)m\over n^2 b} \|\Xh^n-X^n\|^2 &\leq  \bar{d}_o(\Xbbf)+\e',
\end{align}
which, for our set of parameters, proves that, conditioned on $\Ec_1\cap\Ec_2\cap\Ec_3\cap\Ec_4$, ${1\over \sqrt{n}}\|X^n-\Xh^n\|$ can be made arbitrary small. 
\end{proof}

\subsection{Proof of  Theorem \ref{thm:noisycoeffsbound-Dkl}}\label{app:B}

Define event $\Ec_1$ as
\[
\Ec_1=\{{1\over b}D_{\rm KL}(\hat{p}_{k+1}(\cdot| [X_n]_b,q_{k+1})\leq \e\}.
\]
Also define  distributions  $\hat{q}_{k+1}^{(1)}$ and $\hat{q}^{2}_{k+1}$ over $\Xc_b^{k+1}$, as follow: $\hat{q}_{k+1}^{(1)}=\hat{p}_{k+1}(\cdot|[X^n]_b)$
 and $\hat{q}_{k+1}^{(2)}=\hat{p}_{k+1}(\cdot|\Xh^n).$
 Since $\Xh^n$ is the minimizer of $\sum_{a^{k+1}\in\Xc_b^{k+1}} {W}_{a^{k+1}} \hat{p}_{k}(a^{k+1}|u^n) +{\lambda\over n^2}\|Au^n-Y^m\|^2$, we have 
 \begin{align}
&  \sum_{a^{k+1}}W_{a^{k+1}}\hat{p}_{k+1}(a^{k+1}|\Xh^n) +{\lambda\over n^2}\|Y^m-A\Xh^n\|^2\nonumber\\
&\leq \sum_{a^{k+1}}W_{a^{k+1}}\hat{p}_{k+1}(a^{k+1}|[X^n]_b) +{\lambda\over n^2}\|Y^m-A[X^n]_b\|^2.\label{eq:basic-cost}
\end{align}
 On the other hand,  conditioned on $\Ec_1$,  we have
 \begin{align}
  \sum_{a^{k+1}}&W_{a^{k+1}}\hat{p}_{k+1}(a^{k+1}|[X^n]_b) \nonumber\\
  &=- \sum_{a^{k+1}}\hat{q}^{(1)}_{k+1}(a^{k+1})\log q_{k+1}(a_{k+1}|a^k) \nonumber\\  
          &= \sum_{a^{k+1}}\hat{q}^{(1)}_{k+1}(a^{k+1})\log{ \hat{q}^{(1)}_{k+1}(a_{k+1}|a^k)  \over q_{k+1}(a_{k+1}|a^k) \hat{q}^{(1)}_{k+1}(a_{k+1}|a^k)}  \nonumber\\
             &= \hat{H}_k([X^n]_b)+ \sum_{a^{k+1}}\hat{q}^{(1)}_{k+1}(a^{k+1})\log{ \hat{q}^{(1)}_{k+1}(a_{k+1}|a^k)  \over q_{k+1}(a_{k+1}|a^k) }  \nonumber\\ 
                    &= \hat{H}_k([X^n]_b)+ \sum_{a^{k+1}}\hat{q}^{(1)}_{k+1}(a^{k+1})\log{ \hat{q}^{(1)}_{k+1}(a^{k+1})  \over q_{k+1}(a^{k+1}) }\nonumber\\
                    &\;\; \;\; + \sum_{a^{k+1}}\hat{q}^{(1)}_{k+1}(a^{k+1})\log{  q_{k+1}(a^k) \over \hat{q}_{2}(a^k)   }   \nonumber\\
                                     &= \hat{H}_k([X^n]_b)+ D_{\rm KL}(\hat{q}^{(1)}_{k+1}, q_{k+1})  - D_{\rm KL}(\hat{q}^{(1)}_{k}, q_{k})   \nonumber\\ 
                                        &\leq  \hat{H}_k([X^n]_b)+ D_{\rm KL}(\hat{q}^{(1)}_{k+1}, q_{k+1})  \nonumber\\
                                        &\leq  \hat{H}_k([X^n]_b)+ \e b .\label{eq:bound-Xh}
\end{align}
Similarly,
 \begin{align}
  \sum_{a^{k+1}}&W_{a^{k+1}}\hat{p}_{k+1}(a^{k+1}|\Xh^n)\nonumber\\
   &=- \sum_{a^{k+1}}\hat{q}^{(2)}_{k+1}(a^{k+1})\log q_{k+1}(a_{k+1}|a^k) \nonumber\\  
          &= \hat{H}_k(\Xh^n)+ \sum_{a^{k+1}}\hat{q}^{(2)}_{k+1}(a^{k+1})\log{ \hat{q}^{(2)}_{k+1}(a_{k+1}|a^k)  \over q_{k+1}(a_{k+1}|a^k) }  \nonumber\\ 
                    &= \hat{H}_k(\Xh^n)+ \sum_{a^{k}}\hat{q}^{(2)}_{k}(a^{k}) D_{\rm KL}(\hat{q}^{(2)}_{k}(\cdot|a^k), q_{k}(\cdot|a^k))   \nonumber\\ 
                                        &\geq \hat{H}_k(\Xh^n).\label{eq:bound-Xb}
\end{align}
Therefore, combining \eqref{eq:basic-cost}, \eqref{eq:bound-Xh} and \eqref{eq:bound-Xb}, it follows that
  \begin{align}
 \hat{H}_k&(\Xh^n) +{\lambda\over n^2}\|Y^m-A\Xh^n\|^2\nonumber\\
 & \leq  \hat{H}_k([X^n]_b)+ \e b +{\lambda\over n^2}\|Y^m-A[X^n]_b\|^2.\label{eq:B2}
\end{align}
The rest of the proof follows from the proof of Theorem \ref{thm:LMEP}. Note that the proof of Theorem \ref{thm:LMEP} basically starts from the above inequality. , That is, in Theorem \ref{thm:LMEP}, since $\Xh^n$ is the minimizer of  $ \hat{H}_k(u^n)+{\lambda\over n^2}\|Y^m-Au^n\|^2$ over all signals in $\Xc_b^n$, Therefore,
  \begin{align}
 \hat{H}_k&(\Xh^n) +{\lambda\over n^2}\|Y^m-A\Xh^n\|^2\nonumber\\
 & \leq  \hat{H}_k([X^n]_b)+ {\lambda\over n^2}\|Y^m-A[X^n]_b\|^2.\label{eq:B3}
\end{align}
The only difference between \eqref{eq:B2} and \eqref{eq:B3} is the term $\e b$, on the left hand side of \eqref{eq:B2}. However, the effect of this term is asymptotically negligible, as $\e$ is a free parameter that can be made arbitrary small.

\subsection{Proof of  Theorem \ref{cor:nodel-mismatch}}\label{app:C}

  The proof is very similar to the proof of Theorem \ref{thm:noisycoeffsbound}. Define events $\Ec_1$, $\Ec_2$ and $\Ec_3$ as in the proof of Theorem \ref{thm:noisycoeffsbound}. Then, following the same steps as before, conditioned on $\Ec_1\cap\Ec_2\cap\Ec_3$, we have
  \begin{align}\label{eq:UB-H-hat-mismatch-coeffs}
     {1\over b}\hat{H}_{k}([\Xh^n]_b) & \leq \bar{d}_o(\Xbbf)+\e'+\e,
  \end{align}
  where 
  \[
  \e'= 2\e_w +\d+{6\over b}.
  \]
  Unlike  Theorem \ref{thm:noisycoeffsbound}, here $\e_w$ does not go to zero and is a  given  constant. Therefore, we  rewrite \eqref{eq:UB-H-hat-mismatch-coeffs} as
  \begin{align}\label{eq:UB-H-hat-mismatch-coeffs-rewrite}
     {1\over b}\hat{H}_{k}([\Xh^n]_b) & \leq \bar{d}_o(\Xbbf)+3\e_w+\e'',
  \end{align}
  where 
  \[
  \e''\triangleq \d+{6\over b},
  \]
  can be made arbitrary small. Note that $\P(\Ec_1^c\cup\Ec_2^c\cup\Ec_3^c)$ still goes to zero, as $n$ goes to infinity, by the same argument  as before. The rest of the proof follows from the proof of Theorem 5 in   \cite{JalaliP:17-IT}. The difference is that since $\e_w$ does not go to zero in this case, we redefine the set $\Cc_n$ as 
  \[
  \Cc_n=\{u^n\in\Xc_b^n: {1\over nb}\ell_{\rm LZ}(u^n)\leq \bar{d}_o(\Xbbf)+3(\e_w+\d)\}.
  \]
   For $n$ large enough, $\e''\leq 2\d$. Therefore, conditioned on $\Ec_1\cap\Ec_2\cap\Ec_3$, for $n$ large enough,
$[X^n]_b,[\Xh^n]_b\in\Cc_n.$ Note that 
\[
|\Cc_n|\leq 2^{nb(\bar{d}_o(\Xbbf)+3(\e_w+\d))}.
\]
Define event $\Ec_4$ as done in the proof of Theorem \ref{thm:noisycoeffsbound}. The rest of the proof follows similar to the proof of  Theorem \ref{thm:noisycoeffsbound}. The only difference is that since $\e_w$ is fixed, to make sure that $\P(\Ec_4^c)$ converges to zero, almost surely, the number of measurements should be increased from $(1+\d)\bar{d}_o(\Xbbf)n$ to $(1+\d)(\bar{d}_o(\Xbbf)+3\e_w)n$.
   

\section{Conclusion}\label{sec:conclusion}
In this paper we have studied some fundamental  connections between two recently-proposed  compressed sensing recovery methods: i) L-MEP for universal compressed sensing and ii) Q-MAP for Bayesian compressed sensing. 
We have  shown that  a proper approximation of the cost function used in L-MEP yields a variant of the Q-MAP algorithm. The only different between the cost function derived from L-MEP's approximation and the original one used in Q-MAP is that they employ  different weights to promote the desired source structure. This have  motivated us to study the effect of the weights used by the Q-MAP optimization on its performance.  We have shown three different sets of weights for Q-MAP that asymptotically yield  the same performance. In practice, typically, the weights, which are a function of the source distribution,  are to be learned from available training datasets. In such cases,  we have  proved the robustness of the performance of Q-MAP to small error in  estimating the ideal weights. For non-vanishing estimations errors, we have characterized the required increase in the sampling rate to compensate for the error.

\bibliographystyle{unsrt}
\bibliography{myrefs}

\begin{thebibliography}{10}

\bibitem{cover}
T.~Cover and J.~Thomas.
\newblock {\em Elements of Information Theory}.
\newblock Wiley, New York, 2nd edition, 2006.

\bibitem{LZ}
J.~Ziv and A.~Lempel.
\newblock Compression of individual sequences via variable-rate coding.
\newblock {\em Information Theory, IEEE Transactions on}, 24(5):530--536, Sep
  1978.

\bibitem{dude}
T.~Weissman, Erik Ordentlich, G.~Seroussi, S.~Verd\'u, and M.~Weinberger.
\newblock Universal discrete denoising: Known channel.
\newblock {\em IEEE Trans. Inform. Theory}, 51(1):5--28, 2005.

\bibitem{JalaliM:14-MEP-IT}
S.~Jalali, A.~Maleki, and R.~G. Baraniuk.
\newblock Minimum complexity pursuit for universal compressed sensing.
\newblock {\em IEEE Trans. Inform. Theory}, 60(4):2253--2268, Apr. 2014.

\bibitem{BaDu11}
D.~Baron and M.~F. Duarte.
\newblock Universal {MAP} estimation in compressed sensing.
\newblock In {\em 49th Annual Conference on Comm. Control, and Comp.}, pages
  768--775, 2011.

\bibitem{ZhuB:15}
J.~Zhu, D.~Baron, and M.~F. Duarte.
\newblock Recovery from linear measurements with complexity-matching universal
  signal estimation.
\newblock {\em IEEE Trans. Signal Processing}, 63(6):1512--1527, Mar. 2015.

\bibitem{JalaliP:17-IT}
S.~Jalali and H.~V. Poor.
\newblock Universal compressed sensing for almost lossless recovery.
\newblock {\em IEEE Trans. Inform. Theory}, 63(5):2933--2953, May 2017.

\bibitem{beygi2017efficient}
S.~Beygi, S.~Jalali, A.~Maleki, and U.~Mitra.
\newblock An efficient algorithm for compression-based compressed sensing.
\newblock {\em arXiv preprint arXiv:1704.01992}, 2017.

\bibitem{DAMP}
C.~A. Metzler, A.~Maleki, and R.~G. Baraniuk.
\newblock From denoising to compressed sensing.
\newblock {\em IEEE Trans. Inform. Theory}, 62(9):5117--5144, Sep. 2016.

\bibitem{jalali2016qmap}
S.~Jalali and A.~Maleki.
\newblock New approach to bayesian high-dimensional linear regression.
\newblock {\em To appear in Inform. and Inf.}, 2017.

\bibitem{kawabata1994rate}
T.~Kawabata and A.~Dembo.
\newblock The rate-distortion dimension of sets and measures.
\newblock {\em IEEE Trans. Inform. Theory}, 40(5):1564--1572, 1994.

\bibitem{geiger2017information}
B.~C. Geiger and T.~Koch.
\newblock On the information dimension of stochastic processes.
\newblock {\em arXiv preprint arXiv:1702.00645}, 2017.

\bibitem{RezagahJ:17}
F.~E. Rezagah, S.~Jalali, E.~Erkip, and H.~V. Poor.
\newblock Compression-based compressed sensing.
\newblock {\em IEEE Trans. Inform. Theory}, 63(10):6735--6752, Oct. 2017.

\bibitem{renyi1959dimension}
Alfr{\'e}d R{\'e}nyi.
\newblock On the dimension and entropy of probability distributions.
\newblock {\em Acta Math. Acad. Scien. Hungarica}, 10(1-2):193--215, 1959.

\bibitem{shields}
Paul~C. Shields.
\newblock {\em The Ergodic Theory of Discrete Sample Paths}.
\newblock Amer Mathematical Society, July 1996.

\bibitem{JalaliM:12}
S.~Jalali, A.~Montanari, and T.~Weissman.
\newblock Lossy compression of discrete sources via the {Viterbi} algorithm.
\newblock {\em IEEE Trans. Inform. Theory}, 58(4):2475--2489, 2012.

\bibitem{CaTa05}
E.~Cand\`es, J.~Romberg, and T.~Tao.
\newblock Decoding by linear programming.
\newblock {\em IEEE Trans. Inform. Theory}, 51(12):4203 -- 4215, Dec. 2005.

\end{thebibliography}
\end{document}